\documentclass[runningheads]{llncs}

\usepackage{graphicx}
\usepackage{epstopdf}
\usepackage{latexsym,amsmath,amssymb,stmaryrd}
\usepackage{graphicx,epsfig}
\usepackage{datetime}
\usepackage[sort,nocompress]{cite}

\usepackage[usenames,dvipsnames,svgnames,table]{xcolor}
\usepackage{tikz}
\definecolor{darkgreen}{rgb}{0,.5,0}
\newcommand{\cola} [1] {\textcolor{blue}{#1}}
\newcommand{\colb} [1] {\textcolor{darkgreen}{#1}}
\newcommand{\colc} [1] {\textcolor{red}{#1}}
\newcommand{\cold} [1] {\textcolor{brown}{#1}}

\newenvironment{myenumerate}{\begin{enumerate}\setlength{\itemsep}{-.3ex}}%
                             {\end{enumerate}}
                           {\end{itemize}}

\renewcommand{\epsilon}{\varepsilon}

\def\IN{{\mathbb N}}
\def\IR{{\mathbb R}}

\def\pv{{\bf{pv}}}
\def\PV{{\mathcal{{Q}}}} 
 
\def\pdb{PdB}

\def\a{{\tt a}}
\def\b{{\tt b}}
\def\c{{\tt c}}

\def\x{{\tt x}}
\def\y{{\tt y}}

\newtheorem{observation}{Observation}

\title{On the Parikh-de-Bruijn grid}
\author{P{\'e}ter Burcsi\inst{1}
\and Zsuzsanna Lipt{\'a}k\inst{2}
\and W.\ F.\ Smyth\inst{3,4}}

\institute{Dept.\ of Computer Algebra, E\"otv\"os Lor\'and Univ., Budapest, Hungary 
\and 
Dip.\ di Informatica, University of Verona, Verona, Italy
\and
Dept.\ of Computing \& Software, McMaster University, Hamilton, Canada
\and 
School of Engineering \& Inf.\ Techn., Murdoch University, Perth, Australia
}

\begin{document}

\maketitle


\begin{abstract} We introduce the Parikh-de-Bruijn grid, a graph whose vertices are fixed-order Parikh vectors, and whose edges are given by a simple shift operation. This graph gives structural insight into the nature of sets of Parikh vectors as well as that of the Parikh set of a given string. We show its utility by proving some results on Parikh-de-Bruijn strings, the abelian analog of de-Bruijn sequences. 
\end{abstract}


\section{Introduction}

In recent years, interest in so-called abelian stringology has soared. The underlying concept is that of {\em Parikh equivalence} or {\em abelian equivalence}: two strings are Parikh equivalent if the multiplicity of each character is the same in both. Two Parikh equivalent strings can be transformed into one another by permuting their characters; for example, {\tt ENRAGED} and {\tt ANGERED} are Parikh equivalent\footnote{source: {\tt www.fun-with-words.com}}. (See Section~\ref{sec:definitions} for precise definitions.)

Many string problems allow parallel formulations when abelian equivalence is substituted for equality, e.g.\ abelian pattern matching (also known as jumbled pattern matching or histogram matching)~\cite{ButmanEL04,IJFCS12,ToCS12,MoosaR_JDA12,LeeLZ_SPIRE12,TCS14,AmirCLL_ICAPL14,ChanL15,GagieHLW15,ButmanLM16,KociumakaRR17}, 
abelian periods~\cite{CrochemoreIKKPRRTW13,FiciLLP14,FiciLLPS16,KociumakaRR17b}, 
abelian squares, repetitions, and runs~\cite{CS1997,FiciLLLMPP16,FiciKLLP16}, abelian pattern avoidance~\cite{Blanchet-SadriKMSSX12}, abelian borders~\cite{ChristodoulakisCCI14}, 
common abelian factors~\cite{AlatabbiILR16,BadkobehGGNPS16}, 
abelian problems on run-length encoded strings~\cite{IPL13,GiaGrab_IPL13,AmirAHLLR16,CunhaDGWKS17}, 
abelian reconstruction~\cite{AcharyaDMOP15,BarthaBL16}, 
to name just some of the topics that have been treated in recent literature.

In this paper, we present a new model, the Parikh-de-Bruijn grid. The Parikh-de-Bruijn grid is the abelian analogue of the de Bruijn graph; however, as we will see, it has a completely different structure.  

De Bruijn graphs~\cite{deBruijn46} have proved immensely useful in improving our understanding of the exact structure of fixed-sized substring sets of a string, among other things yielding a full description of all strings with the same set or multiset of $k$-length substrings \cite{Ukkonen92,Pevzner95}. This approach gave rise to algorithms for reconstructing strings from their $k$-length substrings, which in turn led to the development of efficient tools for sequence assembly using NGS data (Next-Generation Sequencing)~\cite{PevznerTW01,ZerbinoBirney08}. The impact of these algorithms on  computational biology, and as a consequence, on current research in biology, medicine, and pharmacology, can hardly be overstated: These are the algorithms that have made possible ultra-fast sequence assembly using high-throughput data that has become available with the advent of the new sequencing technologies.

Parikh vectors over an alphabet $\Sigma$ of size $\sigma$ are vectors with $\sigma$ non-negative integer components, and can be viewed as equivalence classes of strings: namely, of Parikh equivalent strings. For example, {\tt aab}, {\tt aba}, {\tt baa} all have Parikh vector $(2,1,0)$ over the alphabet $\Sigma = \{{\tt a}, {\tt b}, {\tt c}\}$. 
The {\em Parikh-de-Bruijn grid} ({\em PdB-grid} for short) is the abelian analog of the de Bruijn graph. Recall that in the de Bruijn graph of order $k$, the vertex set consists of all $k$-length\footnote{There is some inconsistency in the literature as to the {\em order} of a de Bruijn graph: if $V = \Sigma^k$, so that $E=\Sigma^{k+1}$, then some publications refer to this graph as de Bruijn graph of order $k$, and others as order $(k+1)$. Throughout this paper, we call this graph the de Bruijn graph of order $k$.}
strings over the alphabet $\Sigma$, and the edge set of the $(k+1)$-length strings over $\Sigma$.  
In the PdB-grid, the vertices are $k$-order Parikh vectors, i.e.\ Parikh vectors whose coordinates add up to $k$. The edges of the PdB-grid are analogous to the de Bruijn graph in that they represent a simple one-character shift. (For details see  Section~\ref{sec:definitions}.) 

As we will see, the PdB-grid has some properties in common with the de Bruijn graph of the same order; however, there are important differences. As in de Bruijn graphs, every string corresponds to a walk in the PdB-grid; however, the converse is not true: for a walk to correspond to a string, it needs to fulfil additional conditions.  Moreover, $(k+1)$-order Parikh vectors do not correspond to edges, but to certain substructures with $\sigma$ incident edges (triangles in the case of $\sigma = 3$, tetrahedra in the case of $\sigma = 4$, etc.). In contrast to de Bruijn graphs, here 
$(k-1)$-order Parikh vectors, too, correspond to unique substructures. For larger alphabets, the PdB-grid even contains corresponding substructures for certain Parikh vectors of higher and lower orders. In fact, the PdB-grid over alphabet size $\sigma$ can be identified with a regular simplicial complex of dimension ${\sigma -1}$.

The PdB-grid has a different structure from the de Bruijn graph, because the relationships it models, that of Parikh vectors of strings and Parikh vectors of their substrings, is quite different from the substring relationship modeled by de Bruijn graphs. We believe that the PdB-grid may prove to be a powerful tool in improving our understanding of problems in abelian stringology, while giving a common framework for many of the problems in this area. We will demonstrate the utility of the PdB grid by exploring the topic of {\em covering strings} and {\em Parikh-de-Bruijn strings}: the former are strings which contain, for every Parikh vector $p$ of order $k$, a substring with Parikh vector $p$, while the latter contain {\em exactly one} substring for every $k$-order Parikh vector $p$. Parikh-de-Bruijn strings are the abelian analog of the well-known de Bruijn sequences~\cite{deBruijn46}. 

\medskip

The paper is organized as follows: In Section~\ref{sec:definitions}, we give all definitions used in the paper, then we briefly recall de Bruijn graphs, introduce PdB-grids and give first results on PdB-grids and Parikh vector sets. In Section~\ref{sec:covering}, we study $k$-covering and PdB-strings over different alphabets, and ask the question (which we answer partially) for which combinations of $k$ and $\sigma$ PdB-strings exist. 
In Section~\ref{sec:experiments}, we give experimental results about PdB-strings and covering strings for different combinations of $k$ and $\sigma$. 
Finally, in Section~\ref{sec:conclusion}, we give an outlook and state some open problems. 


\section{The Parikh-de-Bruijn grid}\label{sec:definitions}

\subsection{Basic definitions}

Let $\Sigma$ be a finite ordered alphabet with $|\Sigma|=\sigma$. We can write $\Sigma  = \{\a_1 < \ldots < \a_{\sigma}\}$.  
A string (or word) $w=w_1\ldots w_n$ over $\Sigma$ is a finite sequence of characters from $\Sigma$. Its length is denoted by $|w|=n$. For two strings $u,v$ over $\Sigma$, we write $uv$ for their concatenation. If $w=uxv$, for $u,v,x$ (possibly empty) strings, then $u$ is called a prefix, $v$ a suffix, and $x$ a substring or factor of $w$. 

A {\em Parikh vector} (Pv for short) over $\Sigma$ is a vector of length $\sigma$ with non-negative integer entries.  For a string $w$ over $\Sigma$, the Pv of $w$, denoted $\pv(w)$, is defined by $\pv_i = |\{ j \mid w_j = \a_i\}|$, for $i=1,\ldots,\sigma$, the number of occurrences of character $\a_i$ in $w$.  The function $\pv$ induces an equivalence relation on the set of all strings, often called {\em Parikh equivalence} or {\em abelian equivalence}. 
The {\em order} of a Pv $p$ is the sum of its entries, which equals the length of any string $w$ for which\ $p = \pv(w)$. We denote the $i^{\text{th}}$ unit vector by $e_i = \pv(\a_i)$. 

Let $\PV(k,\sigma)$ be the set of all Parikh vectors of order $k$ over an alphabet of size $\sigma$.
It is an easy exercise to show that $|\PV(k,\sigma)| = {k+\sigma - 1 \choose \sigma -1} = {k+\sigma - 1 \choose k}$ (see, e.g.~\cite{Jukna11}). For a string $w$ over $\Sigma$, the {\em Parikh set} of $w$, $\Pi(w)$, is the set of Parikh vectors of all substrings of $w$, and $\Pi_k(w) = \Pi(w) \cap \PV(k,\sigma)$, for $k=1,\ldots,|w|$, the set of Parikh vectors of $k$-length substrings of $w$.
For two Parikh vectors $p$ and $q$ over $\Sigma$, we define: $p \leq q$, if for all $i=1,\ldots, \sigma$, $p_i\leq q_i$;  $p+q = r$ is given by $r_i = p_i+q_i$ for $i=1,\ldots,\sigma$; and, if $p\leq q$, then $q-p=r$ by $r_i = q_i - p_i$, for $i=1,\ldots,\sigma$. For a string $w=uv$, we have $\pv(w) = \pv(u)+\pv(v)$.

\begin{definition}[neighbors, children, parents, meet, join]
Let  $p,p' \in \PV(k,\sigma)$, $q\in \PV(k-1,\sigma)$, and $p_1,\ldots,p_m$ be Parikh vectors of any order. 
\begin{itemize}
\item  $p$ and $p'$ are called {\em neighbors} if there exist $i\neq j$ such that $p' = p - e_i + e_j$. 
\item  $q$ is called a {\em child} of $p$, and $p$ a {\em parent} of $q$, if there exists $i$ such that $p = q + e_i$.  
\item The {\em meet} $\wedge(p_1,\ldots,p_m) = r$ is defined by $r_i = \min \{(p_{j})_i \mid j=1,\ldots,m\}$, and the {\em join} $\vee(p_1,\ldots,p_m) = r'$ by $r'_i = \max \{(p_{j})_i \mid j=1,\ldots,m\}$. (Here $(p_j)_i$ denotes the $i^{\text{th}}$ component of Pv $p_j$.)
\end{itemize}
\end{definition}

Thus $p,p'\in \PV(k,\sigma)$ are neighbors if and only if they can be the Pv's of two consecutive $k$-length windows in some string $w$, i.e.\ if $p = \pv(w_i\ldots w_{i+k-1})$ and $p' = \pv(w_{i+1}\ldots w_{i+k})$ for some $i$. 
Pv $q$ is a child of Pv $p$ if and only if, for some string $w$ such that $\pv(w) = p$, $q$ is the Pv of $w$'s prefix of length $k-1$. 

\begin{example} The neighbors of $(1,2,0)$ are $(0,3,0),(0,2,1),(2,1,0),(1,1,1)$, its parents are 
 $(2,2,0),(1,3,0),(1,2,1)$, and its children are $(0,2,0),(1,1,0)$. 
\end{example}

Observe that every Parikh vector $p$ has $\sigma$ parents and $\gamma$ children, where $\gamma$ equals the number of non-zero components of $p$. 

\begin{lemma}[Clique-Lemma]
\label{lemma:clique}
Let $p_1,p_2,\ldots,p_{\sigma} \in \PV(k,\sigma)$ be pairwise distinct, $\sigma \geq 2$. Let $q = \wedge (p_1,\ldots,p_{\sigma})$ and $r = \vee (p_1,\ldots,p_{\sigma})$. If for all $i\neq j$, $p_i$ is a neighbor of $p_j$, then, for $\sigma>2$, exactly one of the following two cases holds, while for $\sigma=2$, both hold: 
\begin{myenumerate}
\item $\forall i: q \text{ is a child of } p_i,$
\item $\forall i: r \text{ is a parent of } p_i.$
\end{myenumerate}

Conversely, for any $\sigma$, if $q$ is a child of every $p_i$, or if $r$ is a parent of every $p_i$, then the $p_i$ are pairwise neighbors. 
\end{lemma}

\begin{proof}
The converse direction follows easily from the definition. The $\sigma = 2$ case is also immediate. For what follows, notice that the common parent and child of two neighbors is unique. 

For $\sigma \geq 3$ we identify Parikh vectors with the corresponding vertices in the PdB graph. The proof strategy is the following: we first prove that any 3-clique of vertices share a common parent or a common child but not both. Then we show that two adjacent 3-cliques (i.e. sharing 2 vertices) can only form a 4-clique if they are either both child-sharing, or both parent-sharing. This then implies that within an $m$-clique, all 3-cliques are of the same type, completing the proof.

Let the Parikh vectors of a 3-clique be $p$, $q$ and $r$. Then $\exists i, i', j, j'$ s.t. $q = p+e_i-e_j$ and $r=p+e_{i'}-e_{j'}$, $i\neq j$, $i'\neq j'$. The three vertices share a common parent if and only if $j=j'$, and they share a common child if and only if $i=i'$, and these equalities are mutually exclusive. Now suppose neither $i=i'$ nor $j=j'$ holds, then we distinguish three cases. Case 1: $i, i', j, j'$ are pairwise distinct: then $q$ and $r$ differ in 4 coordinates and cannot be neighbors. Case 2: $i=j'$, and $i, i', j$ are pairwise distinct. Then $q$ and $r$ differ in 3 coordinates and are not neighbors. Case 3: $i=j'$, $i'=j$. Then $q = r+2e_i - 2e_j$ and so they are not neighbors. Thus in each case $q$ and $r$ are not neighbors, so that $p$, $q$, $r$ cannot form a clique, a contradiction.

For two adjacent 3-cliques of different type, let the 3-clique with child be $p+e_i, p+e_{i'}, p+e_{i''}$ and the 3-clique with parent be $q-e_j, q-e_{j'}, q-e_{j''}$, and suppose $p+e_i = q-e_j$, $p+e_{i'} = q-e_{j'}$ are the common vertices. Then $e_{i'} + e_{j'} = e_i + e_j$, thus $i=j'$, $i'=j$. Since $i''$ is different from all of $i, i', j, j'$, the vertices $p+e_{i''}$ and $q-e_{j''}$ differ in coordinates $i$, $i'$ and $i''$ and cannot be neighbors.

For an $m$-clique with $m>3$, if all 3-cliques are child-sharing, then this child has to be the same for all vertices of the $m$-clique and we are done. The same is true if all 3-cliques are parent-sharing. If there were both child-sharing and parent sharing 3-cliques within the $m$-clique, then there would also have to be adjacent 3-cliques of different type, which we saw is impossible. \qed

\end{proof}


\subsection{The undirected Parikh-de-Bruijn grid}\label{sec:pdb-grid}

Recall that the de Bruijn graph of order $k$ over alphabet $\Sigma$ is defined as a directed graph $dB(k,\Sigma) = (V,E)$, where $V = \Sigma^k$ the set of all $k$-length strings over $\Sigma$, and $E = \{ (u,v) \mid u_2\ldots u_k = v_1\ldots v_{k-1}\}$. The in-degree and out-degree of every node is $\sigma$, and $E$ corresponds to $\Sigma^{k+1}$. 

We are now ready to define the Parikh-de-Bruijn grid. We will need two different variants, an undirected and a directed one. We start with the undirected variant, and postpone the directed variant until later (Def.~\ref{def:pdb-grid-directed}). 

\def\A{{H}}
\def\G{{G}}

\begin{definition}[Undirected Parikh-de-Bruijn grid]\label{def:pdb-grid-undirected}
Let $k,\sigma\geq 1$. The {\em (undirected) $(k,\sigma)$-Parikh-de-Bruijn grid} ({\em \pdb-grid} for short) is a graph $\A(k,\sigma) =(V,E)$ such that $V = \PV(k,\sigma)$ and $E = \{ pq \mid p \text{ and } q \text{ are neighbors}\}$.\footnote{Here and elsewhere, in slight abuse of notation, we identify Parikh vectors with the vertices representing them.}
\end{definition}

This graph is analogous to the de Bruijn graph in that the edges represent single-character shifts. However, it is undirected, and, as opposed to the de Bruijn graph, the \pdb-grid is not regular. In particular, the degree of a vertex $p$ is $\gamma(\sigma -1)$, where $\gamma$ is the number of non-zero coordinates of $p$. This is because we can choose one character $\a_i$ and a different second character $\a_j$ for a neighbor $q = p - e_i + e_j$. A second important difference to the de Bruijn graph of order $k$ is that there is no correspondence between edges and $(k+1)$-order Pv's. What we do find is a more intricate correspondence, for which we will need a different view of the \pdb-grid.

Let us consider the case $\sigma =3$. Lemma~\ref{lemma:clique} tells us that any triangle in the \pdb-grid of order $k$ will have one common child or one common parent. Consider the drawing of $\A(4,3)$ in Fig.~\ref{fig:pdb43} (left). We find that every triangle which points downward consists of a clique with a common parent, e.g.\ $(2,2,0),(2,1,1),(1,2,1)$, with common parent $(2,2,1)$, while every upward-pointing triangle consists of one with a common child, e.g.\ $(2,2,0),(1,3,0),(1,2,1)$, with common child $(1,2,0)$. We have marked the Pv's in different colors; note that those $(k+1)$-order Pv's which have fewer than $\sigma=3$ children are not enclosed by $3$ vertices, so they cannot be identified with a triangle, but they are drawn as incident to their children (which still form a clique). 

\begin{figure}
\begin{minipage}{6cm}
\begin{center}
\includegraphics[width=\textwidth]{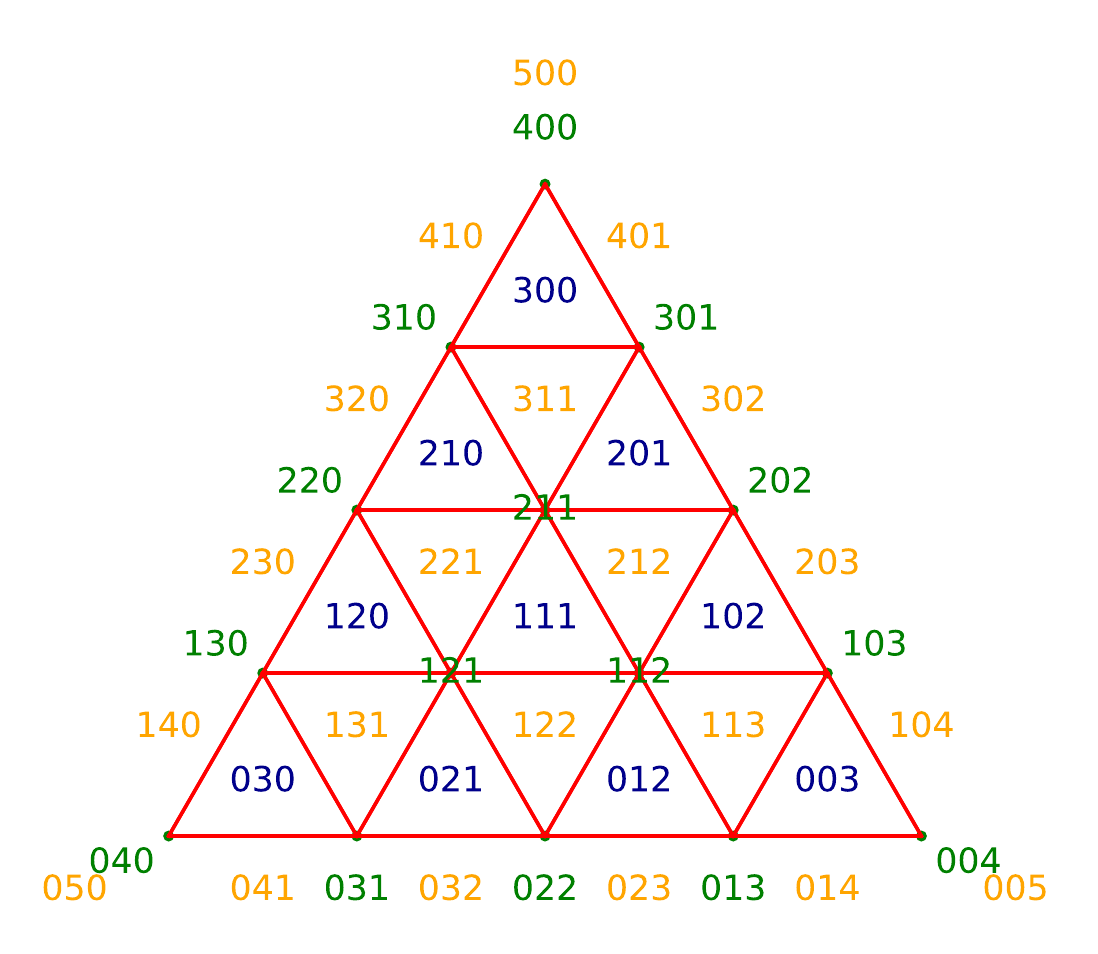}
\end{center}
\end{minipage}
\begin{minipage}{5.8cm}
\begin{center}
\includegraphics[width=\textwidth]{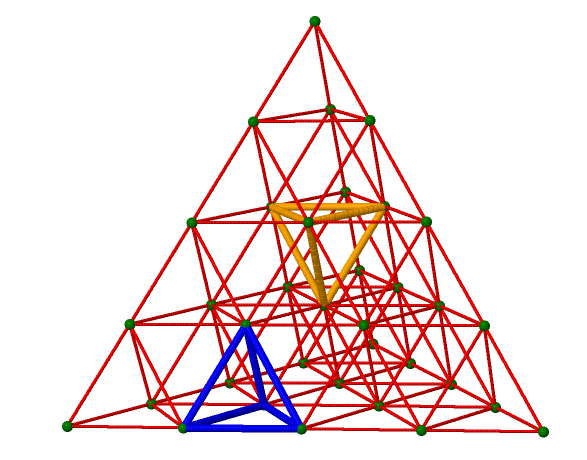}
\end{center}
\end{minipage}
\caption{The \pdb-grid for $k=4$ and $\sigma=3$ (left) resp.\ $\sigma=4$ (right). $k$-order Pv's are denoted in green (vertices), $(k+1)$-order Pv's in yellow (downward triangles resp.\ tetrahedra), and $(k-1)$-order Pv's in blue (upward triangles resp.\ tetrahedra). \label{fig:pdb43}}
\end{figure}

For $\sigma=4$, the $(k+1)$- and $(k-1)$-order Parikh vectors correspond to downward- resp.\ upward-pointing tetrahedra (Fig.~\ref{fig:pdb43}, right). For $\sigma=5$, they will correspond to $5$-cells or pentachora: a pentachoron is a 4-simplex, i.e.\ the convex hull of 5 affinely independent points in $\IR^4$, and has five facets, which are pairwise adjacent tetrahedra, sharing a facet (a triangle) each. 

In general, we can identify the \pdb-grid $\A(k,\sigma)$ with a simplicial complex ${\mathcal S}$ in the following way. Consider the hyperplane ${\mathcal H}$ through the point set $\{k\cdot e_i \mid i=1,\ldots,\sigma\}$. The intersection of ${\mathcal H}$ with $\IN^{\sigma}$, where $\IN$ denotes the set of non-negative integers, consists exactly of the points in $\PV(k,\sigma)$: these are the vertices of our \pdb-grid, and the $0$-simplices of ${\mathcal S}$. The $2$-simplices are given by the edges of the graph, while in addition, we also have $(\sigma-1)$-simplices given by the $\sigma$-cliques in the graph. (By definition, the lower-dimension faces, e.g.\ triangles for $\sigma=4$, are also elements of ${\mathcal S}$, but they are not needed for our purposes. For basic definitions on simplicial complexes, see~\cite{Alexandrov56}). Since all points lie in ${\mathcal H}$, which has dimension $\sigma -1$, we can embed ${\mathcal S}$ in $\IR^{\sigma -1}$. See Fig.~\ref{fig:diagsection} for an illustration. 

\begin{figure}
\begin{center}
\includegraphics[width=.7\textwidth]{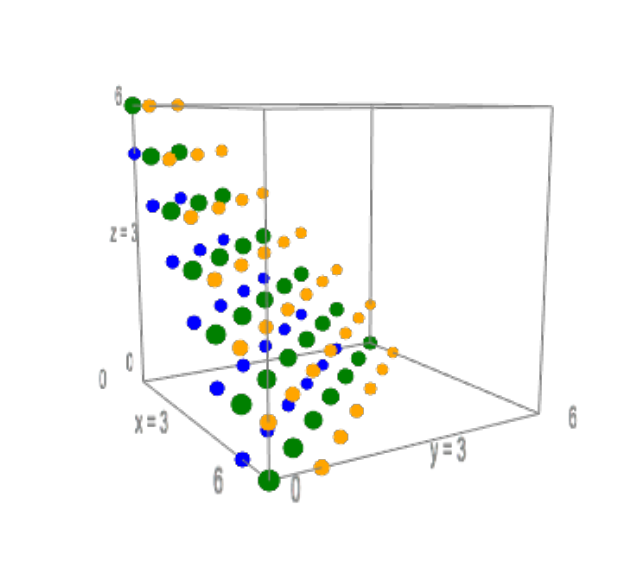}
\end{center}
\caption{The diagonal section of the integer grid with the hyperplanes ${\mathcal H}_k$ (green), ${\mathcal H}_{k+1}$ (blue), and ${\mathcal H}_{k-1}$ (yellow), for $k=6$ and $\sigma=3$. 
\label{fig:diagsection}}
\end{figure}

In the following, we will not distinguish between the \pdb-grid as an abstract graph and the geometric simplicial complex ${\mathcal S}$.


\subsection{The directed Parikh-de-Bruijn grid}\label{sec:pdb-grid2}

We now turn to strings in the \pdb-grid. For this, we need a directed variant of the \pdb-grid which will also include loops. We will call these loops {\em bows} due to the fact that in their geometric realization, they will have a different form from the classic loop. 

\begin{definition}[Labeled Parikh-de-Bruijn grid with loops]\label{def:pdb-grid-directed}
Let $k,\sigma\geq 1$. The {\em directed $(k,\sigma)$-Parikh-de-Bruijn grid} ({\em directed \pdb-grid} for short) is a directed edge-labeled multigraph $\G(k,\sigma) = (V,E)$, where

\begin{itemize}
\item there are two anti-parallel edges between each pair of neighboring vertices $p,q$, where $p = q - e_i + e_j$; edge $(p,q)$ is labeled $(\a_i,\a_j)$, and edge $(q,p)$ is labeled $(\a_j,\a_i)$; 
\item for every vertex $p$, and every non-zero entry $p_i$ of $p$, there is an edge $(p,p)$ labeled $(\a_i, \a_i)$; we call these edges {\em bows}. 
\end{itemize}
\end{definition}

Note that $\A(k,\sigma)$ is the underlying undirected unlabeled simple graph without loops. In $\G(k,\sigma)$, every edge is labeled by the exchange of characters that happens when moving a window of size $k$ along a string from one Pv to the next. When the two Pv's are distinct, there is only one way to do this, hence the unique directed edge from $p\neq q$;  when the two consecutive Pv's are the same, then there are as many ways of doing this as there are characters occurring in it. 

We need a definition which connects strings and walks in $\G(k,\sigma)$, analogously to de Bruijn graphs. Note that we identify a walk in the directed \pdb-grid $\G(k,\sigma)$ with its sequence of vertices, disregarding the edges. This  definition allows using different edges whenever there are consecutive occurrences of the same vertex (since these are the only cases of multiple parallel edges). 

\begin{definition}
Let $w=w_1\ldots w_n$ be a string over $\Sigma$ and $W=(p_1,\ldots,p_m)$ a walk in $\G(k,\sigma)$. Then $W$ {\em spells} $w$ if and only if $m=n-k+1$ and for every $i=1,\ldots,m$, $\pv(w_i\ldots w_{i+k-1}) = p_i$. In particular, if $W$ spells $w$, then, for every $i=1,\ldots,m-1$, there is an edge $(p_i,p_{i+1})$ with label $(w_i,w_{i+k})$. A walk is called {\em realizable} if it spells some string.
\end{definition}

\begin{example}
Let $\sigma = 3, k=4$. Consider the string ${\tt aabacabb}$. The walk it induces in the $(4,3)$-\pdb-grid is shown in Fig.~\ref{fig:ex2}. Notice that each shift by one in the string corresponds to a step along an edge. If this edge connects two distinct vertices, then the $(k+1)$- and $(k-1)$-order Pv's are given by the two triangles incident to the edge. This is the case in the first and last step in our example (marked in blue, resp.\ brown). If the edge connects a vertex $p$ with itself, then the induced $(k+1)$- and $(k-1)$-order Pv's are given by triangles at the {\em opposite} direction w.r.t.\ $p$. This is the case in our example in the second (red) and third (green) steps. 
\end{example}

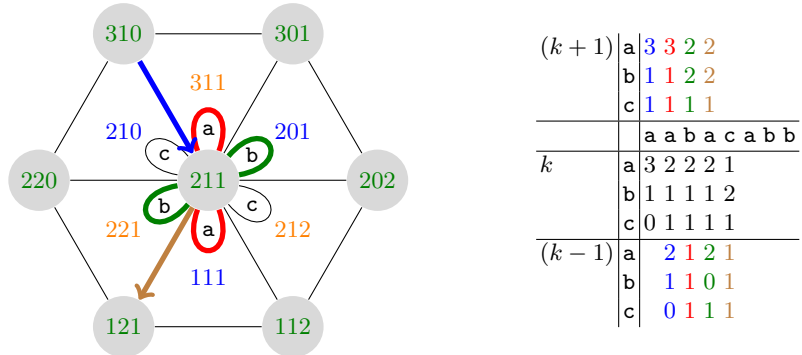
\begin{figure}
\begin{minipage}{6cm}
\begin{center}
\begin{tikzpicture}[scale=.9]
\tikzstyle{every node} = [circle, fill=gray!30, text=darkgreen]
\def\elength{2.5}

\node (v) at (0, 0) {211};
\node (b0) at +(0: \elength) {202};
\node (b1) at +(60: \elength) {301};
\node (b2) [thick] at +(120: \elength) {310};
\node (b3) at +(180: \elength) {220};
\node (b4) at +(240: \elength) {121};
\node (b5) at +(300: \elength) {112};

\draw [-] (b0) -- (b1);
\draw [-] (b1) -- (b2);
\draw [-] (b2) -- (b3);
\draw [-] (b3) -- (b4);
\draw [-] (b5) -- (b0);
\draw [-] (b4) -- (b5);

\def\loose{7}

\draw [-,darkgreen,line width=.7mm,out=10,in=50,looseness=\loose] (v) to (v);
\draw [-,red,line width=.7mm,out=70,in=110,looseness=\loose] (v) to [out=70,in=110,looseness=\loose] (v);
\draw (v) to [out=130,in=170,looseness=\loose] (v);
\draw [-,darkgreen,line width=.7mm,out=190,in=230,looseness=\loose] (v) to (v);
\draw [-,red,line width=.7mm,out=250,in=290,looseness=\loose] (v) to (v);
\draw (v) to [out=310,in=350,looseness=\loose] (v);

\draw [-] (v) -- (b0);
\draw [-] (v) -- (b1);
\draw [<-,blue,line width=.7mm] (v) -- (b2);
\draw [-] (v) -- (b3);
\draw [->,brown,line width=.7mm] (v) -- (b4);
\draw [-] (v) -- (b5);

\tikzstyle{every node}=[align=center, text=black]

\node at +(30: .3*\elength) {\tt{b}};
\node at +(150: .3*\elength) {\tt{c}};
\node at +(270: .3*\elength) {\tt{a}};

\node at +(210: .3*\elength) {\tt{b}};
\node at +(330: .3*\elength) {\tt{c}};
\node at +(90: .3*\elength) {\tt{a}};

\tikzstyle{every node}=[align=center, text=blue]

\node at +(30: .577*\elength) {201};
\node at +(150: .577*\elength) {210};
\node at +(270: .577*\elength) {111};

\tikzstyle{every node}=[align=center, text=orange]
\node at +(90: .577*\elength) {311};
\node at +(210: .577*\elength) {221};
\node at +(330: .577*\elength) {212};

\end{tikzpicture}
\end{center}
\end{minipage}
\begin{minipage}{6cm}
\begin{center}
\begin{tabular}{l | l | c c c c c c c c}
$(k+1)$ & \a & \cola 3  & \colc 3 & \colb 2 & \cold 2\\
 &\b & \cola 1  & \colc 1 & \colb 2 & \cold 2\\
 &\c & \cola 1  & \colc 1 & \colb 1 & \cold 1\\
 \hline
 && \a & \a & \b & \a & \c & \a & \b & \b \\
 \hline
$k$ &\a & 3 & 2 & 2 & 2 & 1\\
 & \b & 1 & 1 & 1 & 1 & 2\\
 &\c & 0 & 1 & 1 & 1 & 1\\
\hline
$(k-1)$ &\a& & \cola 2 & \colc 1 & \colb 2 & \cold 1\\
 & \b& & \cola 1 & \colc 1 & \colb 0 & \cold 1\\
 &\c& & \cola 0 & \colc 1 & \colb 1 & \cold 1\\
\end{tabular}
\end{center}
\end{minipage}

\caption{Left: The part of the $(4,3)$-\pdb-grid being visited by the string ${\tt aabacabb}$. We marked each step of the walk in a different colour (blue-red-green-brown). Right: We list vertically the corresponding $(k+1)$-order Pv's (top), $k$-order Pv's (center), and $(k-1)$-order Pv's, each at its beginning position. 
\label{fig:ex2}}
\end{figure}
Notice in particular that the sequence of edges used by the walk in the \pdb-grid of order $k$ determines not only the vertices touched (the $k$-order Pv's), but also which $(k+1)$- and $(k-1)$-order Pv's are visited: those are the ones which are incident to the edges of the walk. For higher $\sigma$, even lower and higher order Pv's have corresponding substructures. For example, for $\sigma=4$, the octahedra which are enclosed by 6 $k$-order Pv's correspond to their common ''grandchildren''. 

\begin{proposition}
Let $w$ be a string over $\Sigma$ and $k\geq 1$. Then there is a walk in $\G(k,\sigma)$ which spells $w$. On the other hand, not every walk spells a string. 
\end{proposition}

\begin{proof} If $w$ is a string, then the definition of neighbor means that there is a walk in $\G(k,\sigma)$. To see that the other direction does not hold, consider the walk $(3,0,0),(2,1,0),(3,0,0)$, which does not spell any string. \qed
\end{proof}

The following characterization of realizable walks follows directly from the definitions. As we will see, it restricts very strongly which walks are realizable and which are not. 

\begin{theorem}[Characterization of realizable walks]
A walk $W=(p_1,\ldots,p_m)$ in $\G(k,\sigma)$ is realizable if and only if, for all $i=1,\ldots, m-k-1$, 
there is a character $\c_i$ s.t.\ there is an edge labeled $(\x,\c_i)$ from $p_i$ to $p_{i+1}$ and an edge labeled $(\c_i,\y)$ from $p_{i+k}$ to $p_{i+k+1}$, where $\x$ and $\y$ are arbitrary characters from $\Sigma$. 
\end{theorem}

The following lemma concerns walks that do not use bows, i.e.\ walks for which $p_i \neq p_{i+1}$ for all $i$: 

\begin{lemma}[Bowfree walks]\label{lemma:bowfree}
Let $W = (p_1,\ldots, p_n)$ be a realizable walk in $\G(k,\sigma)$ which does not use any bows. 
  \begin{enumerate}
\item  Then, if $w=w_1\ldots w_{n+k-1}$ is a string spelled by $W,$ then for every $i=1,\ldots, n-1$, $w_i \neq w_{i+k}$. 
\item If $p_i = k\cdot e_j$ for some $j$, then, for $q=p_{i+k}$, it follows that $q_j = 0$, i.e.\ $q$ lies on the opposite face of $\G(k,\sigma)$ to the corner $k\cdot e_j$. 
\end{enumerate}
\end{lemma}

\begin{proof} {\em 1.} Let $\x = w_i = w_{i+k}$. Then $p_i= \pv(w_i\ldots w_{i+k-1}) = \pv(w_{i+1}\ldots w_{i+k}) = p_{i+1}$, so the bow at $p_i$ labeled with character $\x$ is used by the walk $W$. 

{\em 2.} Let $w$ be a word spelled by $W$. Since $p_i = k\cdot e_j$, it follows that $w_i\ldots w_{i+k-1} = {\tt a_j}^k$. Assume that $w_{i+k}\ldots w_{i+2k-1}$ contains a character $\a_j$, say at position $i+k+\ell$.  Then $p_{i+\ell} = p_{i+\ell +1}$, in contradiction to $W$ not using any bows. \qed
\end{proof}


\subsection{Realizable sets}\label{sec:realizable}

\begin{definition}[realizable sets]
A set $\Pi$ of $k$-order Parikh vectors is called {\em realizable} if there is a string $s$ such that $\Pi_k(s) = \Pi$. 
\end{definition}

Clearly, every singleton set is realizable. 

\begin{lemma}\label{lemma:neighbors}
Let $\Pi = \{p,q\}$. Then $\Pi$ is realizable if and only if $p$ and $q$ are neighbors. 
\end{lemma} 

\begin{proof}
If $\Pi$ is realizable, then there is a string $w$ such that $\Pi_k(w) = \{p,q\}$. Then there must be a position $\ell$ in $w$ such that, w.l.o.g.\ $\pv(w_{\ell}\ldots w_{\ell+k-1}) = p$ and $\pv(w_{\ell+1}\ldots w_{\ell+k}) = q$. Let $w_{\ell} = \a_i$ and $w_{\ell+k} = \a_j$. Since $p\neq q$, we have $\a_i\neq \a_j$ and $q = p - e_i + e_j$. Conversely, let $q = p - e_i + e_j$, $\wedge(p,q) = r$, and $t$ any string with $\pv(t)=r$. Then the string $\a_it\a_j$ realizes $\Pi$. \qed
\end{proof}

\begin{theorem}\label{thm:connected}
A set of $k$-order Parikh vectors $\Pi$ is realizable if and only if the induced subgraph $G(\Pi)$ of the \pdb-grid $\G(k,\sigma)$ is connected. 
\end{theorem}

\def\iti{{\textit iti}}

\begin{proof}
Given a walk $W$, denote by $\iti(W)$ the bowfree walk obtained by replacing all consecutive multiple occurrences of the same vertex by a single occurrence (the {\em itinerary of $W$}). We will show that for any bowfree walk $U$, there is a walk $W$ and a string $w$ such that $W$ spells $w$ and $U = \iti(W)$. This proves the theorem, since in a connected induced subgraph there always exists a bowfree walk covering all vertices of this subgraph.

The proof goes by induction on the length $n$ of $U$, the base case $n=1$ being trivial. Now suppose a string $u$ exists with itinerary $(p_1, \ldots p_{n-1})$. Thus the final $k$ letters of $u=u_1u_2\cdots u_m$ have Parikh vector $p_{n-1}$. Let $i$ and $j$ be such that $p_n = p_{n-1} + e_i -e_j$, so the suffix of length $k$ of $u$ contains $\a_j$, say $u_{m-g} = \a_j$ with $0\leq g\leq k-1$. Append the string $u_{m-k+1}u_{m-k+2}\cdots u_{m-g-1} \a_i$ to $u$. The new Parikh vectors obtained this way are several more copies of $p_{n-1}$ (using bows by repeating the letters from $k$ positions before) and a $p_n$ in the last suffix, giving the desired itinerary. \qed
\end{proof}

Next we turn to strings which realize all of $\PV(k,\sigma)$. 

\begin{definition}[$k$-covering strings, $k$-\pdb-strings]
Let $k, \sigma$ be positive integers and let $w$ be a string over $\Sigma$ with $|\Sigma|=\sigma$. The string $w$ is called 
\begin{itemize}
\item {\em $(k,\sigma)$-covering}, 
if, for every Parikh vector $p$ of order $k$, there is a substring $u$ of $w$ such that $\pv(u) = p$, and

\item {\em $(k,\sigma)$-\pdb-string} ({\em Parikh-de-Bruijn string of order $k$}) if, for every Pv $p$, there is exactly one substring $u$ of $w$ such that $\pv(u) = p$. 
\end{itemize}

\noindent
If $\sigma$ is clear from the context, we also say $k$-covering and $k$-PdB word.
\end{definition}

\pdb-strings are the abelian equivalent of {\em de Bruijn sequences of order $k$}: these are strings $s$ with the property that every string of length $k$ appears exactly once as a substring of $s$. De Bruijn sequences exist for every combination of $k$ and $\sigma$, since there is a one-to-one correspondence between de Bruijn sequences of order $k$ and Hamiltonian paths in the de Bruijn graph of order $k$; and they can be constructed efficiently using the fact that Eulerian paths in the $(k-1)$-order de Bruijn graph correspond to de Bruijn sequences of order $k$. 

Every \pdb-string of order $k$ corresponds to  a Hamiltonian path in the \pdb-grid. However, the converse is not true, because not every Hamiltonian path is realizable. In general, \pdb-strings do not exist for every pair $(k,\sigma)$. In the next section, we study some necessary and sufficient conditions for their existence. 


\section{Covering words and Parikh de Bruijn words}\label{sec:covering}
\label{sec:pdb-strings}

Below, we will use the convention that the first letters of the alphabet $\Sigma$ are $\a, \b, \c, \ldots$, unless otherwise stated.  

\begin{observation}
\label{obs:pdblength}
A $(k, \sigma)$-covering word has length at least $\binom{\sigma+k-1}{k} + k-1$, with equality if and only if the word is a PdB word. 
\end{observation}

This holds because a word of length $n$ contains at most $n-k+1$ substrings of length $k$, and thus at most $n-k+1$ Parikh vectors of order $k$. For a $(k,\sigma)$-covering word $w$ of length $n$, we refer to $n - (\binom{\sigma+k-1}{k} + k-1)$ as the {\em excess} of $w$. For $\sigma=2$, the words of the form $\a^k\b^k$ are \pdb-words, and indeed they  have excess $0$. 

\medskip

A classical de Bruijn word contains all possible strings of length $k$ and thus necessarily all possible strings of shorter lengths: $k-1$, $k-2$ etc. This is not always true for Parikh de Bruijn words: for example, the word 
$$w={\tt aaaaabbbbbcaaaadbbbcccccdddddaaaccdbcbaccaccddbddbadacddbbbb}$$

\noindent is a $(5, 4)$-PdB word that is not $4$-covering: it has no substring with Pv $(1,1,1,1)$. The following two statements show how the $k$-covering property relates to the $(k-1)$-covering property. 

\begin{proposition}
\label{prop:covering1}
If $\sigma\leq 2$ or $k\leq 3$, then a $k$-covering word is also $(k-1)$-covering.
\end{proposition}

\begin{proof}
Let $w$ denote the $k$-covering word. If $\sigma = 2$, then $(0, k), (k,0)\in \Pi_k(w)$, meaning $\a^k$ and $\b^k$ occur in $w$, implying that $\a^{k-1}$ and $\b^{k-1}$ also occur in $w$ and thus $(0, k-1), (k-1,0) \in \Pi_{k-1}(w)$. By the convexity property\footnote{Also referred to as interval property or continuity property: For a binary string $w$, if for some $x<y<k$, $(x,k-x), (y,k-y) \in \Pi(w)$, then also $(z,k-z)\in \Pi(w)$ for every $x\leq z \leq y$. Rediscovered many times, can be referred to as ``folklore''.} of Parikh sets over a 2-letter alphabet, $\Pi_{k-1}(w)$ contains all possible Parikh vectors of order $k-1$.

For general alphabets, the case of $k=2$ is trivial. If $k=3$, then by symmetry, it is enough to show that $\a^2$ occurs in $w$ and that either $\a\b$ or $\b\a$ occurs in $w$. Since $w$ is $3$-covering, $\a^3$ occurs in it of which $\a^2$ is a factor. Also, one of ${\tt aab}$, ${\tt aba}$, ${\tt baa}$ occurs, meaning that a neighboring $\a$ and $\b$ are present. \qed
\end{proof}

\begin{proposition}
\label{prop:kcoverbut}
If $\sigma \geq 3$ and $k\geq 4$, then there exist words which are $k$-covering, but are not $(k-1)$-covering.
\end{proposition}

\begin{proof}
We will show that a $k$-covering word exists that avoids the Parikh vector $p=(k-3, 1, 1, \underbrace{0,\ldots,0}_{\text{$\sigma-3$ zeroes}})$. 
Consider the $k$-PdB grid, in which $p$ corresponds to a simplex whose vertices are the parents of $p$. Our goal is to construct a string $w$ so that the walk in the grid corresponding to $w$ visits every vertex in the grid, but avoids edges of the simplex representing $p$. Note that this includes in particular bows corresponding to $p$. We will do this by first constructing a walk that avoids even the vertices of this simplex, and then modifying it so that it visits these vertices, but does so in a ``safe'' way: whenever the walk visits a vertex of this simplex, in the next step it immediately turns away from it. In particular, it does not use any bows incident to $p$. 

If we remove the parent vertices of $p$ (together with incident edges and bows) from the PdB grid, then the remaining graph is still connected and contains all bows in all remaining vertices, thus by Theorem~\ref{thm:connected} there exists a word $w'$ that visits the remaining vertices. This word is almost $k$-covering, except for the vertices of the removed simplex. We now show how to modify $w'$ so that it visits all vertices of the grid without using edges of the simplex corresponding to $p$. Observe that for each $\x\in \Sigma$, $w'$ already contains $\x^k$ as a factor. We obtain a new string $w$ from $w'$ by replacing, for every $\x\in\Sigma$, one factor $\x^k$ by the string $\x^k u(\x) \x^k$, specified below for each $\x\in\Sigma$. This way it is ensured that any new Parikh vector can only appear within these replacement strings.

Notice that the simplex corresponding to $p$ lies on a side of the PdB-grid. The vertices of this simplex can be categorized into three types: (1) the parent where the first coordinate is increased as compared to $p$, (2) the two parents where the second or third coordinate is increased, and (3) those where any other coordinate is increased. We will define $u(\x)$ differently for the three types of parents, namely: 

\begin{enumerate}
\item $u(\a) =\b\a^{k-2}\c$; 
\item $u(\b)=\a^{k-3}\b^2\c$ and $u(\c)=\a^{k-3}\c^2\b$; 
\item $u(\x)=\b\a^{k-3}\x\c$ for all $\x\in \Sigma \setminus \{\a, \b, \c\}$. 
\end{enumerate}

It is straightforward to check that the word $w$ obtained in this way really visits the parents of $p$ in the desired manner. \qed

\end{proof}

This only implies that at least one Parikh vector of order $k-1$ can be avoided. One could ask to what extent a $k$-covering word can fail to be $(k-1)$-covering.

\newcommand{\covprop}{\textrm{mincov}}

\begin{question}
\label{question:mincov}
Consider the minimum proportion of visited $(k-1)$-order Parikh vectors among $k$-covering words, that is
\begin{equation}
\covprop(k, \sigma) =  \min \left\{\frac{|\Pi_{k-1}(w)|}{\binom{\sigma + k -2}{k-1}} \mid w \textrm{ is } (k, \sigma) \textrm{-covering} \right\}
\end{equation}
Can we determine the exact value of $\liminf_{k\to\infty}\covprop(k, \sigma)$? Is the $\liminf$ actually a limit? 
\end{question}

\newcommand{\covset}{\textrm{covset}}

\begin{question}
Consider $\covset(w) = \{ k \mid w \textrm{ is } k\textrm{-covering} \}$.
Which finite sets $K\subseteq \mathbb{N}^+$ are realizable as $K = \covset(w)$ for some $w\in \Sigma^*$? For $\sigma = 2$ these are exactly sets of the form $\{1, 2, \ldots, k\}$. For $\sigma \geq 3$, we have by Proposition \ref{prop:covering1} that two necessary conditions are $3\in K \Rightarrow 2\in K$ and $2\in K \Rightarrow 1\in K$. It seems possible that these conditions are also sufficient, but a simple proof of this has eluded us so far.
\end{question}

\subsection{Shortest covering words}\label{sec:shortest}

We now turn our attention to shortest $k$-covering words.  

The classic concept of a {\em universal cycle} is defined as follows~\cite{CDG92}: Given a set ${\cal F}_n$ of combinatorial objects with the property that each object can be represented by a (not necessarily unique) string of length $n$ over alphabet $A$; a {\em universal cycle for ${\cal F}_n$} is a cyclic string over $A$ of length $m$ s.t.\ there is a one-to-one correspondence between elements of ${\cal F}_n$ and substrings $a_i\ldots a_{i+n-1}$ of length $n$ of $A$, where addition is modulo $m$. In particular, necessarily $m=|{\cal F}_n|$. 

Recall that classical de Bruijn sequences come in two flavours, linear or cyclic sequences. 
One can similarly define {\em cyclic \pdb-strings}: these have exactly one substring for each Pv of order $k$, where the last $k-1$ positions are viewed as continuing with the beginning of the string.
Thus, cyclic de Bruijn sequences of order $k$ are universal cycles for $\Sigma^k$, while cyclic $(k,\sigma)$-\pdb-words are universal cycles for $\PV(k,\sigma)$. It is easy to see that there is a one-to-one correspondence between $\PV(k,\sigma)$ and $k$-size multisets over a ground set of size $\sigma$ (also called multicombinations), therefore cyclic \pdb-strings can also be viewed as universal cycles for multicombinations.
 There are some known results about these and closely related universal cycles~\cite{CDG92,J93,HJZ09,SBEIMSVVYZ02,J_manus,TAOCP4}.

\medskip

A universal cycle viewed as a cyclic string $w$ always gives rise to a PdB word: simply append the first $k-1$ letters of $w$ at the end. But PdB words can also be constructed that do not come from such a cycle: an example is ${\tt abbbcccaaabc}$, which is a PdB word for $\sigma = k = 3$. Theorem~\ref{thm:chung} below implies that in fact no $(3, 3)$-PdB word comes from a universal cycle.

The following necessary condition is known from the literature on universal cycles. The proof, translated into the language of Parikh vectors, goes by counting occurrences of letters and the contribution of each occurrence to individual coordinates of Parikh vectors. 

\begin{theorem}\cite{CDG92}
\label{thm:chung}
If a universal cycle on $k$ multisets exists, then $\binom{\sigma+k-1}{k-1} / k$ is an integer.
\end{theorem}

Slightly modifying the argument from the proof in~\cite{CDG92}, we obtain some lower bounds for the length of $k$-covering words.

\begin{theorem}
\label{thm:shortest}
A shortest $k$-covering word has length at least
\begin{equation}
\max \left(\binom{\sigma+k-1}{k} + k-1, \sigma \cdot \left\lceil \binom{\sigma+k-1}{k-1} / k \right\rceil \right)
\end{equation} 

\noindent In particular, a $(k,\sigma)$-\pdb-word can only exist if 
$\binom{\sigma+k-1}{k} + k-1 \geq \sigma \left\lceil \binom{\sigma+k-1}{k-1} / k \right\rceil$. 
\end{theorem}

\begin{proof} 
Consider occurrences of the letter $\a$. Count all pairs $(j,
p)$ where the letter at position $j$ is $\a$ and is covered by a
substring of length $k$ with Parikh vector exactly $p$.  Each occurrence of $\a$
is counted at most $k$ times, thus the number of such pairs is at most $kn_\a$,
where $n_\a$ is the number of occurrences of $\a$ in the string. If the
word is covering, then all possible $\binom {\sigma+k-1}{k}$ Parikh
vectors occur somewhere, each covering some occurrences (possibly 0)
of letters $\a$. The number of $\a$'s in all Parikh vectors is $k\cdot\binom
 {\sigma+k-1}{k}/ \sigma = \binom{\sigma+k-1}{k-1}$. From this we have

\begin{equation}
k n_a \geq \binom{\sigma+k-1}{k-1} 
\end{equation}

and so 
\begin{equation}
\label{eq:1}
n_\a \geq \left \lceil  \binom{\sigma+k-1}{k-1}/ k \right\rceil
\end{equation}

Summing for all letters in the alphabet, we get 
\[
|w| = \sum_{\x\in \Sigma} n_{\x} \geq \sigma \cdot \left\lceil \binom{\sigma+k-1}{k-1} / k \right\rceil
\]
The other lower bound comes from Observation \ref{obs:pdblength}. \qed
\end{proof}

\begin{corollary}
\label{cor:shortest2}
If $\sigma > k^2 - k$ then a $(k, \sigma)$-PdB word can exist only if $k$ divides $\binom{\sigma + k - 1}{k-1}$. As a special case, if $\sigma > k^2 - k$ and $k$ is a prime, then a $(k, \sigma)$-PdB word can exist only if 
$k$ does not divide $\sigma$.
\end{corollary}

\begin{proof} 
If $\binom{\sigma+k-1}{k-1}/ k$ fails to be an integer, then rounding it increases its value by at least $1/k$. Thus
\begin{eqnarray*}
\sigma \cdot \left\lceil \binom{\sigma+k-1}{k-1} / k \right\rceil &\geq& \frac{\sigma}{k}\cdot  \binom{\sigma+k-1}{k-1} + \frac{\sigma}{k} = \\ 
\binom{\sigma + k -1}{k} + \frac{\sigma}{k} &>& \binom{\sigma + k -1}{k} + k - 1
\end{eqnarray*}
if $\sigma>k^2-k$. 

Note that if $k$ is prime, then $\binom{\sigma+k-1}{k-1}/ k$ is an integer if and only if $k$ divides one of the $(k-1)$ neighboring integers $(\sigma+1)(\sigma+2)\cdots (\sigma+k-1)$. This is the case if and only if $k$ does not divide $\sigma$.

\qed
\end{proof}

We now give upper bounds for the length of shortest covering words.

\begin{proposition}
\label{prop:coverk2}
If $k=2$ and $\sigma \leq 3$ then the shortest $k$-covering string has length $\binom{\sigma + 1}{2} + 1$ if $\sigma$ is odd, and $\binom{\sigma + 1}{2} + \sigma/2$ for $\sigma$ even. This means that for odd $\sigma$, PdB words always exist, and for even $\sigma$ the minimal excess is $\sigma/2-1$.
\end{proposition}

\begin{proof}
We use the language of PdB-grids, but an equivalent elementary reasoning is also possible. Consider the $k-1$ grid: this is essentially a complete graph on $\sigma$ vertices. The $k$-simplices corresponding to order-$k$ Parikh vectors are all outside this grid now, only represented by edges for Parikh vectors containing two $1$s, and loops for Parikh vectors containing a $2$. In order to have a $k$-covering string, we have to visit all of these simplices, that is, we have to find a walk covering all edges in a graph that consists of the complete graph on $\sigma$ vertices, $K_\sigma$, together with a loop at each vertex. If $\sigma$ is odd, then a Eulerian walk exists, if $\sigma$ is even, then we have to traverse at least $\sigma/2 - 1$ edges twice, but this suffices. \qed
\end{proof}

\begin{proposition}
\label{prop:pdbsigma3}
For $k=3$, PdB words exist if and only if either $\sigma = 3$ or $\sigma$ is not a multiple of $3$.
\end{proposition}

\begin{proof}
Constructions for universal cycles for $\sigma$ not a multiple of 3 are given in~\cite{J93,HJZ09}. A $(3, 3)$ PdB word is given by ${\tt abbbcccaaabc}$ (in fact up to reversal and permutations of $\Sigma$, this is the only one as one can verify by a computer search). If $\sigma\geq 6$ is a multiple of $3$, then Corollary \ref{cor:shortest2} applies. \qed
\end{proof}

\begin{question}
What is the length of the shortest covering strings for $\sigma\geq 6$ and $k=3$ if $\sigma$ is a multiple of $3$?
\end{question}

For general $k$ and large alphabets, it was conjectured in \cite{CDG92} that 
for $\sigma$ large enough, depending on $k$, universal cycles exist if and only if $k$ divides $\binom{\sigma+k-1}{k-1}$.
By Corollary \ref{cor:shortest2}, this would imply the following statement.

\begin{conjecture}
For $\sigma$ large enough depending on $k$, \pdb-words exist if and only if $k$ divides $\binom{\sigma+k-1}{k-1}$.
\end{conjecture}

From personal communication \cite{J_manus}, we know that the cases $k=4$ and $k=5$ of these conjectures have been settled.
A final result in this section gives an impossibility result for $\sigma = 3$. 

\begin{theorem}
\label{thm:nopdb}
No PdB words exist for $\sigma = 3$ and $k\geq 4$.
\end{theorem}

\begin{proof} 
Suppose a PdB word $w$ exists. Consider the walk $p_1, p_2, \ldots, p_{\binom{\sigma+k-1}{k}}$ corresponding to $w$ in the $(k, 3)$-PdB grid. By symmetry, we may assume that $(k, 0, 0)$,  $(0, k, 0)$ and $(0, 0, k)$ occur in the walk in this order, say $p_x = (k, 0, 0)$, $p_y=(0, k, 0)$, $p_z = (0, 0, k)$ with $x<y<z$. We know by Lemma~\ref{lemma:bowfree} that for every position $i$, $w_i \neq w_{i+k}$, since otherwise there would be bows in the walk contradicting the PdB property. Thus on the $k$ positions before and after the occurrence of $\b^k$, we only have letters $\a$ and $\c$, in other words $p_{y-k}$ and $p_{y+k}$ both lie on the side of the triangular grid opposite $(0, k, 0)$. Similarly, $p_{x+k}$ lies on the side opposite $(k, 0, 0)$, and $p_{z-k}$ on the side opposite to $(0, 0, k)$.

The path of length $k$ from $p_{y-k}$ to $p_y$ disconnects the grid into two regions. The path from $p_x$ to $p_{x+k}$ lies entirely in one of these regions.
But the only point in this region lying on the side opposite $(k, 0, 0)$ is $(0, k, 0),$  meaning that $p_{x+k} = p_y$, so $y=x+k$. A similar argument shows $z=y+k$. The string $w$ thus contains the factor $\a^k\b^k\c^k$. If $k\geq 5$, then at least $k$ letters precede $\a^k$ or at least $k$ letters follow $\c^k$; w.l.o.g.\ assume the latter. This implies that $p_{z+k}$ lies on the side opposite $(0, 0, k)$, but all of those vertices have already been visited, a contradiction. The case $k=4$ can be excluded by a computer search, or a slightly more intricate analysis of the possible letters preceding and following $\a^4\b^4\c^4$. 
\qed 
\end{proof}

We note that Theorem \ref{thm:nopdb} is in line with a conjecture made by Donald Knuth in ''The Art of Computer Programming''~\cite{TAOCP4} (Sec.\ 7.2.1.3, Problem 109). The conjecture states that a universal cycle on $k$-multisets of a set of size $\sigma$ exists if and only if a universal cycle on $k$-sets of a set of size $\sigma+k-1$ exists. If $\sigma=3$, then this would mean a cycle of $k$-sets of a $(k+2)$ sized set, and it was proved in \cite{SBEIMSVVYZ02} that such cycles do not exist.

\section{Experimental results}\label{sec:experiments}

We searched for shortest covering words for moderate values of $\sigma$ and $k$. In some cases, the search was performed via a backtracking brute force algorithm. In other cases (esp. for larger values of $\sigma$ with $k=3$), for each Parikh vector $p$ of order $k$, we carefully planned a candidate factor from $\Sigma^k$ that realizes $p$. The ideas we used were similar to the ones appearing in \cite{J93,HJZ09}. We give some of our results in Table \ref{tab:minwords}. Recall that the excess of a word $w$ equals $|w| - (\binom{\sigma+k-1}{k} + k-1)$. 

\noindent
\begin{table}
\begin{center}
\begin{tabular}{|c|c|l|c|c|}
\hline
$\sigma$ & $k$ & word & length & PdB (excess) \\
\hline
3 & 2 & \tt{aabbcca} & 7 & yes \\
\hline
3 & 3 & \tt{abbbcccaaabc} & 12 & yes \\
\hline
3 & 4 & \tt{aaaabbbbccccaacabcb} & 19 & no (1) \\
\hline
3 & 5 & \tt{aaaaabbbacccccbbbbbaacaaccb} & 27 & no (2) \\
\hline
3 & 6 & \tt{aaaabccccccaaaaaabbbbbbcccbbcabbaca} & 35 & no (2) \\
\hline
3 & 7 & \tt{aabbbccbbcccabacaaabcbbbbbbbaaaaaaacccccccba} & 44 & no (2) \\
\hline
4 & 2 & \tt{aabbcadbccdd} & 12 & no (1) \\
\hline
4 & 3 & \tt{aaabbbcaadbdbccadddccc} & 22 & yes \\
\hline
4 & 4 & \tt{aabbbbcaacadbddbccacddddaaaabdbbccccdd} & 38 & yes \\
\hline
4 & 5 & \scriptsize{\tt{aaaaabbbbbcaaaadbbbcccccdddddaaaccdbcbaccaccddbddbadacddbbbb}} & 60 & yes \\
\hline
5 & 2 & \tt{aabbcadbeccddeea} & 16 & yes \\
\hline
5 & 3 & \tt{aaabbbcaadbbeaccbdddcccebededadceeeaa} & 37 & yes \\
\hline
5 & 4 & \scriptsize{\tt{aaaabbbbcaaadbbbeaaccbbddaaeaebcccadbeeeadddcccceeeeddddbebecbdcdeceacdad}} & 73 & yes \\
\hline
\end{tabular}

\vspace{,5cm}
\caption{Examples of shortest $k$-covering words for various values of $\sigma$ and $k$.}
\label{tab:minwords}
\end{center}
\end{table}

\section{Conclusion}\label{sec:conclusion}

In this paper we introduced the Parikh-de-Bruijn grid, which we believe could prove a powerful tool in abelian stringology. In accordance with the differences between abelian string problems and their classical counterparts, the \pdb-grid has very different properties from the classical de Bruijn graph. 
We studied realizable Parikh sets, covering words, and Parikh de Bruijn words. Further directions to explore include the connection between combinatorial properties of the walk of a string in the grid and properties of the string itself: can we characterize abelian squares, abelian borders and other concepts using the graph structure or the geometry of the \pdb-grid?

\begin{small}

\end{small}

\end{document}